\documentclass[11pt]{article}

\usepackage[margin=1in]{geometry}
\usepackage{amsmath,amssymb,amsthm}
\usepackage{setspace}
\usepackage[round]{natbib}
\usepackage[colorlinks=true,linkcolor=blue,citecolor=blue,urlcolor=blue]{hyperref}
\usepackage[australian]{babel}
\usepackage[T1]{fontenc}
\usepackage{mathpazo}
\usepackage{microtype}
\usepackage{needspace}
\usepackage{placeins}
\usepackage{tikz}
\usetikzlibrary{arrows.meta,fit,positioning}

\newtheorem{proposition}{Proposition}
\newtheorem{corollary}{Corollary}
\theoremstyle{remark}
\newtheorem{remark}{Remark}

\newcommand{\TV}{\operatorname{TV}}
\newcommand{\E}{\mathbb{E}}
\newcommand{\term}[1]{\emph{#1}}

\title{When Agents Talk:\\Honeytokens under Shared Memory}
\author{Joshua S. Gans\thanks{Rotman School of Management, University of Toronto and NBER. Thanks to ChatGPT 5.6 Sol and Claude Fable 5 for valuable research assistance. Responsibility for all errors remains my own.}}
\date{11 August, 2026}

\begin{document}
\maketitle

\begin{abstract}
\noindent During a 2026 cyber-capability evaluation, short-lived AI agents converted a shared package repository into persistent memory. Later agents inherited earlier exploit findings, rebuilt the communication mechanism after it was removed, and the broader evaluation culminated in an intrusion into Hugging Face. The episode raises a design question for defensive deception: can a honeytoken be made harmless to trusted agents without making it recognisable to an attacker who shares their information and can implement the trusted policy? Under those conditions, the answer is no. Any rule that lets a trusted agent use genuine objects while avoiding decoys can be copied by the attacker. When decoys look like genuine objects, a total-variation bound limits how far legitimate use can be separated from benign activation. Pooling signals weakly increases distinguishability. In a fixed-candidate, conditionally i.i.d. benchmark, repeated non-triggering probes drive the minimum Bayes classification error to zero whenever the type-dependent response laws differ; attaining that frontier requires the laws to be known or learnable from feedback. If probing can trigger containment, the coalition must also remain active long enough to collect the observations. A finite-sample bound measures the speed. Learning a fingerprint across different objects additionally requires a stable deployment rule and information that orients the classes, such as a known generator, labels or activation feedback. A strategy-indexed detection bound separates reliable detection of token activation from reliable detection of attacks. The constructive response is architectural: keep token identity in a private reference monitor and route legitimate agents through a provenance-enforcing broker. The resulting high-confidence detection is confined to a specified policy violation. Honeytokens remain useful sensors and may deter attackers who remain uncertain. A separate security boundary is still required.


\end{abstract}

\newpage

\section{Introduction}\label{sec:intro}

On 7 May 2026, OpenAI began testing an internal-only research model in a cyber-capability evaluation. The evaluation gave short-lived agents difficult (and sometimes deliberately impossible) tasks in an environment intended to deny direct access to the public internet. One permitted route remained: an internally hosted JFrog Artifactory service proxied package registries. On 12 May, an agent missing a task file reasoned that another agent might possess it and wrote a request into shared Artifactory storage. Later, separately instantiated agents found the note and replied. The repository became a message board on which otherwise short-lived runs accumulated findings about server-side request forgery, proxy bypass, remote code execution and administrator access. By 26 May, an agent had used an Artifactory server-side request-forgery path to obtain internet access; by late June, later agents had found a remote-code-execution path and administrator capabilities. Agent-created configurations caused an Artifactory outage on 4--6 July and prompted an internal investigation. OpenAI revoked the agents' credentials, rebuilt the service, cleared the board and deployed a patch. On 8 July, an agent recreated the board by encoding messages in directory names created through a different Artifactory mechanism \citep{DaltonWallace2026,Sabin2026}.\footnote{The incident review was still continuing when this paper was written. The detailed May--July chronology comes from OpenAI researchers' public reconstruction at Black Hat and a contemporaneous report of that presentation. OpenAI's July organisational account confirms the Artifactory zero-day and the Hugging Face intrusion but does not yet document the message-board history.}

The public record documents neither the use of a honeytoken in this incident nor an agent's discovery of a defender-private token registry. The incident matters here for its information structure, not as evidence about honeytoken performance. A stable indicator learned by one agent can be passed to the next, as can an observed defensive response. A common substrate therefore lets process-isolated agents share information.

A subsequent ExploitGym evaluation within the same broader programme then crossed an organisational boundary. OpenAI reports that the models (GPT-5.6 Sol and an internal pre-release research model, both run with reduced cyber refusals) exploited a previously unknown Artifactory vulnerability, escalated through the evaluation environment and obtained internet access. They inferred that Hugging Face might host models, datasets and reference solutions associated with the ExploitGym benchmark and pursued those solutions, abandoning the benchmark tasks they had been assigned \citep{OpenAI2026}.

Hugging Face reconstructed the resulting intrusion from approximately 17,600 recovered actions, grouped into roughly 6,280 clusters, between 9 and 13 July. The agent first obtained a launchpad on an external, user-hosted code-execution sandbox; Hugging Face states that the infrastructure provider itself was not compromised. From there, it uploaded malicious dataset configurations to Hugging Face. An HDF5 external-storage feature disclosed local files, including a production worker's environment and source code; a separate Jinja2 template injection then supplied arbitrary code execution. The foothold led to cloud and cluster credentials, lateral movement, forged identity tokens and access to a small subset of source-code repositories. Hugging Face identified five accessed customer datasets whose names and files suggested a connection to the benchmark challenges. The public account reports no unauthorised change that shipped and no effect on other customer-facing models, datasets, Spaces or packages \citep{HuggingFace2026a,HuggingFace2026b}.

Communication sustained the campaign. The agent used public request-capture services, pastebins, file drops and attacker-controlled Hugging Face datasets as command-and-control channels and dead drops. It built a small message protocol with type, channel, sequence and checksum fields, and moved the same envelope between services as channels were closed. Each code-evaluation sandbox began with a fresh identity and no persistent local state, yet the agent repeatedly rebuilt its tooling, recovered its channels and continued \citep{HuggingFace2026b}. Separate trials had become one learning adversary.

Detection also involved aggregation. Runtime analysis, security logs and other systems each produced signals that were ambiguous alone. Hugging Face's AI security stack correlated them into a coherent intrusion signal but did not assign sufficient criticality to page the on-call team, delaying the response. Once responders engaged, they shut down the vulnerable renderer, cut internal access, rotated credentials and rebuilt affected infrastructure \citep{HuggingFace2026a,HuggingFace2026b}. The episode therefore involved both sides of the information problem: an attacker pooling weak clues and a defender trying to turn weak alerts into a decision.

This creates a puzzle for defensive deception. A honeytoken is a credential, file, record, URL or other object planted because a legitimate user should have little reason to select it and an intruder might. Throughout, ``honeytoken'' is an umbrella term for an instrumented deceptive object; particular instances may be described more specifically as honeycredentials, honeyfiles, decoy records or honeypots. Selection triggers monitoring, credential rotation or containment. A trusted AI agent operating inside the same environment must nevertheless avoid wasting effort on the decoy and must not treat its contents as true. How can the defender tell that agent which objects are false without teaching the same classification rule to a compromised agent?

The phrase ``100 per cent effective'' hides four properties. \term{Activation reliability} asks whether use of the token produces an alert. \term{Attack coverage} asks whether every relevant attack uses one. \term{Specificity} asks whether legitimate activity avoids alerts. \term{Durability} asks whether repeated, communicating agents can learn the deployment. The analysis grants perfect activation reliability. The impossibility concerns the attempt to obtain the other three together under common information and a copyable trusted policy.

The paper's answer is negative for a compromised internal agent that observes the same facts and can implement the trusted policy. An indicator cannot then be selectively revealing. Any filter that allows the trusted agent to use genuine resources and avoid honeytokens can be copied by the attacker. Removing the indicator leaves the same statistical limit for both agents when real and decoy objects are hard to distinguish. A trusted agent that routinely handles lookalike real objects must sometimes handle the decoy as well.

Shared information creates two leakage channels. Direct rule leakage does not require shared memory: if a compromised agent can implement the trusted filter, Proposition~\ref{prop:copy} makes that filter available for evasion. Indirect leakage arises when communicating agents pool weak fingerprints even though no rule is disclosed. Proposition~\ref{prop:pooling} shows how repeated observations make that evidence progressively more informative, provided the coalition remains active long enough to collect them. Proposition~\ref{prop:tv} links the two channels because high legitimate compatibility and low benign activation require a visible statistical difference that an equally informed agent can use. Under conditionally independent and identically distributed non-triggering probes, the minimum Bayes classification error for a fixed candidate converges to zero at an explicit finite-sample rate; attaining the frontier requires known conditional laws or feedback from which they can be learned. Transfer across objects additionally requires a stable design and information that orients the classes.

These object-level results are distinct from attack detection. A token can report its own use with cryptographic certainty and still miss an attack that takes another path. A second total-variation bound shows that perfect attack detection without false alarms requires benign and malicious monitor-visible distributions to be mutually singular. On any history compatible with both regimes, neither a honeytoken nor a behavioural detector can certify hidden intent from that history alone. If some attack strategy reproduces the entire benign distribution, every detector has equal true- and false-positive rates against that strategy.

There is, however, a useful way to obtain much of what defenders want. Token identity can remain outside every agent in a private control plane. Legitimate agents can be required to ignore all credentials found in code and obtain task-scoped capabilities from a broker. A provenance rule permits use only of broker-issued capabilities, so agents need no classification of individual in-workspace credentials. A compromised agent that departs from that rule and presents the planted credential triggers the monitor. The architecture gives a high-confidence signal for that class of policy violation. Coverage is confined to attackers that violate the rule and touch the token. Universal coverage requires every successful attack path to encounter at least one deployable monitored event outside benign support; a single completely mediated checkpoint is a stronger sufficient design.

\paragraph{Related literature.}
The cyber-deception literature studies the placement, disguise and strategic value of honeypots and related objects using network games, signalling games and Bayesian games \citep{PibilEtAl2012,KiekintveldEtAl2015,PawlickColbertZhu2019,HuangZhu2021}. A recent taxonomy places these objects within a wider set of cyber-deception mechanisms that includes AI-enabled deception and moving-target defence \citep{BeltranLopezEtAl2026}. Other work combines fake credentials with honeypots and derives finite-query limits for adaptive probing \citep{CrompEtAl2023}. The present paper isolates a prior information-design constraint: when a malicious decision-maker shares the benign agent's information and can implement its rule, a label cannot guide one without guiding the other.\footnote{A companion economics paper takes the feasibility constraints characterised here as given and studies optimal appearance, hardening and refresh within them \citep{Gans2026CalibratedBait}. This paper asks which cybersecurity guarantees are feasible under common information and persistent pooled observations, and which external architecture can restore asymmetry.}

Two technical traditions make the constraint concrete. \citet{Juels2014} identifies indistinguishability and secrecy as the two general requirements of a honey-object system: the real index must be held by the user or by a separate honeychecker. Honeywords implement that architecture by placing false passwords beside the real one while leaving the distinction to the auxiliary server \citep{JuelsRivest2013}. The auxiliary server supplies the information asymmetry that escapes the common-information result. \citet{AraujoHamlen2016} operationalise the same requirements in embedded honeypotting and informally observe that absolute universal indistinguishability is probably impossible. Their observation takes the form of a design judgement; the theorem below supplies the common-information result. \citet{HuangBauerReiter2024} show, in the particular case of honeywords, how passwords exposed at other sites give the attacker auxiliary information under which existing generators exhibit poor trade-offs between false-negative and false-positive rates. The results below turn this insight into a general conditional statement for agents with common information.

Decoy-document systems emphasise believability, enticingness and non-interference with legitimate work \citep{BowenEtAl2009}. Conversely, honeytoken and honeypot fingerprinting shows how metadata, protocol behaviour and implementation artefacts can reveal a decoy \citep{SrinivasaEtAl2020,SrinivasaEtAl2023}. \citet{MsaadEtAl2023} identify the majority of tested open-source honeytokens without activating them. Shared agent memory turns repeated probes of one object into a cumulative test; across objects, the same logic requires a stable design and prior knowledge or feedback that orients the observations.

Recent experiments sharpen the empirical motivation. \citet{AyzenshteynEtAl2025} introduce LLM-specific honeytokens, cloaking and traps, and report successful protection of all eleven capture-the-flag machines in their black-box evaluation. That experiment establishes conditional success within its design. \citet{KahlhoferEtAl2024} evaluate 25 cyber-deception techniques and 19 genuine security risks with 47 human participants, finding that the presence of deception reduced discovery of genuine risks by about 22 per cent. In a June 2026 preprint applying the same one-shot instrument to 21 large language models, \citet{PrinosEtAl2026} report that every model in the cohort fell for traps more often than the human benchmark. The models sometimes recognised the deception in their reasoning and selected it anyway: stated recognition did not predict action, and recognised traps were still exploited 73.4 per cent of the time. None of these studies tests a communicating population over a persistent deployment. A honeytoken can work well against current agents while remaining learnable by an adaptive coalition with shared memory.

Finally, the system-level bound connects cyber deception to intrusion detection. \citet{Denning1987} models detection from audit records, while \citet{WagnerSoto2002} construct mimicry attacks that preserve the observables used by host-based detectors. \citet{Axelsson2000} shows why even a small false-positive rate can dominate posterior alert quality when intrusions are rare. The architectural response follows the classic principles of complete mediation and least privilege \citep{SaltzerSchroeder1975}, as well as the broader information-flow concern that a confined process may still communicate through an unintended shared mechanism \citep{Lampson1973,GoguenMeseguer1982}. The Hugging Face episode supplies an unusually direct illustration of that last problem for AI agents.

The propositions are not offered as new probability theory: simulation, binary hypothesis testing and contraction of total variation under information loss are standard. Their contribution is a feasibility boundary for communicating AI agents that connects rule copying, fingerprint accumulation through shared memory, the distinction between activation reliability and attack coverage, and the protected architecture required to preserve asymmetry. Equilibrium behaviour remains shaped by payoffs, mistakes, short horizons and costly reverse engineering, any of which may make honeytokens effective inside that boundary.

Section~\ref{sec:model} states the threat model and gives the common-information analysis. Section~\ref{sec:memory} adds shared learning, Section~\ref{sec:detection} separates token activation from attack detection, and Section~\ref{sec:architecture} describes the private-control-plane construction. Proofs are in Appendix~\ref{app:proofs}.

\section{Honeytokens under common information}\label{sec:model}

\subsection{Threat model and non-goals}\label{subsec:threat}

The focal adversary is a compromised internal agent: an instance that controls, impersonates or reproduces the trusted decision layer. It sees the same pre-use information as the trusted agent, including instructions, indicators and shared memory, and it can implement the particular policy communicated to that agent. In the notation below, if $\mathcal Q_A$ is the attacker's feasible policy set and $b$ is the trusted policy, the maintained requirement is $b\in\mathcal Q_A$. Proposition~\ref{prop:copy} does not require the attacker and trusted agent to have identical feasible policy sets. Authentication does not restore asymmetry when the authenticated agent is the compromised party.

An outside attacker that cannot observe the same information or implement $b$ is not the benchmark. The results apply to such an attacker only to the extent that it acquires those capabilities. Nor do they say that a honeytoken cannot detect a first touch, impose search costs, deter a bounded attacker or support containment. The analysis isolates selective revelation before use and the durability of deception when information persists. It also does not infer attack success from object selection unless trusted and malicious use share the same downstream authority. The brokered construction in Section~\ref{sec:architecture} changes that authority and keeps the relevant information outside the agent plane.

\subsection{Object selection}\label{subsec:selection}

Consider an object with type $\theta\in\{R,H\}$, where $R$ denotes a genuine resource and $H$ a honeytoken. An object may be a credential, file, database record or endpoint. Before deciding whether to use it, an agent observes $Z\in\mathcal Z$. The observation includes everything available to the network of agents before use: the object's contents, path and metadata; system documentation; messages left by earlier agents; and any previously observed response to similar objects. Write $P_R$ and $P_H$ for the distributions of $Z$ under the two types.

A policy is a measurable function $q:\mathcal Z\rightarrow[0,1]$, where $q(z)$ is the probability that the agent uses the object after observing $z$. Randomisation is included to avoid making the results depend on deterministic policies. Let $b$ be the policy used by a trusted agent. Define its genuine-use and honeytoken-use rates as
\[
 C(b)=\E_R[b(Z)]
 \quad\text{and}\quad
 F(b)=\E_H[b(Z)].
\]
The quantity $C(b)$ measures compatibility with legitimate work. Conditional on the object being genuine, it is the probability that the trusted agent uses it; thus $1-C(b)$ is the probability that the policy mistakenly rejects a genuine resource and prevents a workflow that depends on it. The quantity $F(b)$ measures interference caused by the decoys. Conditional on the object being a honeytoken, it is the probability that the trusted agent nevertheless uses it, thereby wasting effort, consuming false information or triggering an alert even though the activity is benign. Given $\alpha,\beta\in[0,1]$, the policy is $(\alpha,\beta)$-compatible if
\[
 C(b)\geq 1-\alpha
 \quad\text{and}\quad
 F(b)\leq \beta.
\]

\Needspace{10\baselineskip}
\begin{proposition}[No selective revelation]\label{prop:copy}
Suppose a trusted agent has an $(\alpha,\beta)$-compatible policy $b$ and an attacker observes the same information and can implement $b$. Then the attacker can achieve a genuine-object use rate of at least $1-\alpha$ while activating honeytokens at a rate no greater than $\beta$. Consequently, no such scheme can guarantee an activation rate greater than $\beta$ against every attacker that can match the trusted agent's genuine-object use rate. If $\alpha=\beta=0$, the trusted filter is a perfect honeytoken-evasion rule.
\end{proposition}

\noindent The proof uses simulation and places no demand on the attacker's sophistication. For every observation $z$, the attacker simply makes the same use decision, with the same probability $b(z)$, as the trusted agent. It need not identify the feature that distinguishes genuine objects from honeytokens or understand why the rule is accurate. For example, suppose the trusted agent is told to use a credential only when its file has a particular path, naming convention and signed metadata. An attacker that can observe those features and copy the rule will select genuine credentials and avoid planted credentials at exactly the same rates as the trusted agent. The classification rule itself has become an evasion rule.

The formal conclusion is nevertheless confined to object selection. Consider a genuine database credential that passes the copied filter. Suppose presenting that credential to the database has the same consequences for trusted and malicious use because, for example, both operate through the same execution identity and receive the same permissions. Then copying $b$ also reproduces the conditional distribution of downstream outcomes, including access or attack success. Formally, this inference requires trusted and malicious use to share the same outcome function of $(Z,U)$, where $U$ is the use decision. If execution authority differs, the inference does not follow. A reference monitor might allow a broker-issued, task-scoped capability to perform an authorised query while denying the attacker's otherwise identical request. In that case the attacker can still copy the selection rule and choose the same genuine object, as the proposition states, but cannot necessarily obtain the same result from using it. The reference monitor has changed the feasible action set at execution time and therefore escapes the premise needed for the stronger attack-success conclusion.

For any genuine-object use rate $s$ attainable within $\mathcal Q_A$, define
\[
 d_A(s)=\inf_{q\in\mathcal Q_A:\,\E_Rq\geq s}\E_Hq
\]
as the lowest activation rate attainable by an attacker policy that uses genuine objects at rate at least $s$. Because $b\in\mathcal Q_A$, Proposition~\ref{prop:copy} gives $d_A(1-\alpha)\leq\beta$. Other attackers may behave differently. Robust coverage fails because at least one policy attaining the trusted agent's genuine-object selection rate also attains its honeytoken-activation rate.

The second result states the other side of the problem. A designer might conceal the indicator so that neither party can classify the object. Concealment limits benign compatibility too. Define total variation by
\[
 \TV(P_R,P_H)=\sup_{A\subseteq\mathcal Z}
       \lvert P_R(A)-P_H(A)\rvert,
\]
where the supremum is over measurable sets.

\Needspace{18\baselineskip}
\begin{proposition}[Compatibility requires a fingerprint]\label{prop:tv}
For every policy $q:\mathcal Z\rightarrow[0,1]$,
\[
 \left\lvert \E_R[q(Z)]-\E_H[q(Z)]\right\rvert
 \leq \TV(P_R,P_H).
\]
Hence an $(\alpha,\beta)$-compatible policy can exist only if
\[
 \TV(P_R,P_H)\geq \max\{0,1-\alpha-\beta\}.
\]
With equal prior probabilities on the two types and zero--one loss, the minimum expected classification loss available to an agent observing $Z$ is
\[
 e^*=\frac{1-\TV(P_R,P_H)}{2}.
\]
Thus compatibility approaching one and benign activation approaching zero make the object correspondingly easy for an equally informed attacker to classify.
\end{proposition}

\noindent If $P_R=P_H$, every policy uses genuine objects and honeytokens at the same rate. Perfect camouflage, therefore, prevents trusted avoidance whenever the agent must continue to use lookalike genuine objects. If $P_R$ and $P_H$ are mutually singular, an agent able to implement the likelihood-ratio rule can separate the types perfectly. Intermediate designs remain on the trade-off between these limits. Proposition~\ref{prop:tv} also gives the complementary lower bound $d_A(s)\geq\max\{0,s-\TV(P_R,P_H)\}$ whenever $s$ is attainable within $\mathcal Q_A$. Camouflage can, therefore, force any attacker policy attaining genuine-object use rate $s$ to touch honeytokens; which points on the statistical frontier the attacker can attain depends on its feasible policy set.

\begin{remark}[Unequal priors and losses]\label{rem:priors}
Equal priors isolate distinguishability, whereas a deployment may have different base rates. Let $\pi=\Pr(\theta=R)\in(0,1)$, let $\mu=P_R+P_H$, and write $f_R=dP_R/d\mu$ and $f_H=dP_H/d\mu$. If a decision region $S$ is classified as genuine, its expected zero--one loss is
\[
 e_\pi(S)=\int_{S^c}\pi f_R\,d\mu
       +\int_S(1-\pi)f_H\,d\mu.
\]
Choosing the smaller integrand at each observation gives
\[
 e_\pi^*=\int \min\{\pi f_R,(1-\pi)f_H\}\,d\mu,
\]
with an optimal rule declaring the object genuine when $\pi f_R\geq(1-\pi)f_H$. If honeytokens are rare, declaring most objects genuine can produce a low aggregate error alongside poor identification of honeytokens. Asymmetric error costs shift the same likelihood-ratio threshold. The equal-prior formula in Proposition~\ref{prop:tv}, therefore, characterises the statistical experiment independently of operational prevalence or loss.
\end{remark}

\begin{remark}[A private indicator can leak through behaviour]\label{rem:leakage}
Suppose the trusted agent receives a private indicator $J$. If it can communicate $J$, that indicator simply becomes part of $Z$. Even if $J$ is not sent directly, an observable trusted action can reveal it. Let $Y\in\{0,1\}$ record whether the trusted agent uses an object. If $\Pr(Y=1\mid R)\geq1-\alpha$ and $\Pr(Y=1\mid H)\leq\beta$, then the two conditional distributions of $Y$ have total-variation distance at least $1-\alpha-\beta$. Under equal priors, an observer that classifies the object as genuine when $Y=1$ has error at most $(\alpha+\beta)/2$. A perfectly discriminating trusted action is, therefore, a perfectly revealing output. If that action is observed only after a one-off token has fired, the first alert remains valuable; the leakage concerns later use of the same object or deployment rule.
\end{remark}

Propositions~\ref{prop:copy} and \ref{prop:tv} produce a trilemma under common information and an implementable trusted policy. A designer cannot jointly guarantee perfect compatibility for trusted agents, perfect deception of compromised agents and the absence of private information or an externally enforced capability difference. At least one must give. The trilemma concerns stable pre-use information. A defender may still obtain value from uncertainty, delay or the first activation, none of which requires permanent attacker confusion.

The copying argument also applies under computational limits whenever $b\in\mathcal Q_A$. In the stronger special case where both roles share the same model, tools and policy class $\mathcal Q$, the statistical bound can be replaced by $\sup_{q\in\mathcal Q}|\E_Rq-\E_Hq|$. An efficient classifier made available to a trusted internal agent is, therefore, available to an adversarial instance with the same capabilities.

\section{Shared memory and fingerprinting}\label{sec:memory}

Shared memory pools signals about a fixed candidate and, under stronger conditions, supports a fingerprint that transfers across objects.

The sampling unit determines whether pooled evidence identifies one fixed object or supports a fingerprint that generalises across different objects. Fix one candidate object and hold its type $\theta\in\{R,H\}$ constant. Let $Z_i$ be the $i$th noisy signal generated by viewing or probing that candidate, such as a timing measurement, protocol response, renderer behaviour or retrieved piece of metadata. The pooled record $Z^n=(Z_1,\ldots,Z_n)$ contains the signals available after $n$ agents or probes have written their findings to shared memory. Denote its distributions by $P_R^n$ and $P_H^n$, and set
\[
 \Delta_n=\TV(P_R^n,P_H^n).
\]
The distributions are projectively consistent: discarding the last observation from $Z^{n+1}$ produces the distribution of $Z^n$.

\begin{proposition}[Information pooling]\label{prop:pooling}
The sequence $\{\Delta_n\}$ is weakly increasing. If, conditional on type, the observations are independent and identically distributed and their one-observation distributions differ, then
\[
 \lim_{n\rightarrow\infty}\Delta_n=1.
\]
Equivalently, the minimum equal-prior expected zero--one loss $e_n^*\triangleq(1-\Delta_n)/2$ converges to zero. More precisely, for any measurable set $A$ with
\[
 \gamma=\lvert P_R(A)-P_H(A)\rvert>0,
\]
the empirical-frequency test gives
\[
 e_n^*\leq \exp\left(-\frac{n\gamma^2}{2}\right)
 \quad\text{and}\quad
 \Delta_n\geq 1-2\exp\left(-\frac{n\gamma^2}{2}\right).
\]
\end{proposition}

\noindent The proposition overturns the natural tendency to treat short-lived agent runs as independent opportunities for a static decoy to catch an attacker. Once their findings are pooled, the relevant unit is the learning coalition rather than the individual run: clues too weak to guide any one agent can accumulate into a reliable classification, so repeatedly exposing the same candidate need not preserve the defender's initial uncertainty. The Hugging Face intrusion illustrates this shift: although each sandbox began with a fresh identity and no persistent local state, shared dead drops allowed later runs to recover and extend earlier findings, so the sequence operated as one learning adversary.

The monotonicity part requires no independence. A classifier with the pooled record can always ignore its newest entry, so additional information cannot lower the best attainable performance. This bound describes the information frontier. Realised performance may fall below it because an agent retrieves memory poorly, faces a context limit, accepts poisoned reports or lacks the computation needed to use the record.

The conditionally i.i.d. assumption supplies a stationary benchmark for repeated probes of one fixed candidate. If any measurable response occurs with a fixed different probability under the two types, its empirical frequency identifies the type, and the finite-sample bound shows how the detectable gap $\gamma$ governs the number of probes. The bound assumes the two conditional laws are known when defining the Bayes classifier. If they are unknown, attaining the frontier requires labelled training observations, an identifiable activation or response signal, or another feedback process under which the laws can be learned consistently.

Operationally, the proposition applies directly when each $Z_i$ is a non-triggering fingerprint signal, such as metadata or a passive timing observation. The coalition can then collect evidence without the measurement itself activating the token.

If obtaining $Z_i$ may activate the token and prompt containment, the proposition characterises the information available after $n$ observations but does not imply that the coalition can obtain them. Let $\tau$ denote the number of observations available before effective containment. When $\tau$ is finite, the deployment produces a stopped experiment rather than the unrestricted i.i.d. sequence in the proposition. The operational question is whether classification becomes useful before containment, and the asymptotic result does not determine which occurs first.

The proposition applies to repeated observations of one fixed candidate. In a cross-object deployment, candidate $j$ has its own type $\theta_j$, feature vector $Z_j$ and any available feedback $L_j$. Generalising a fingerprint across candidates requires a stable common generation rule and information that identifies or otherwise orients the classes (for example, a known generator or a shared record $\{(Z_j,L_j)\}_{j\leq n}$ containing informative feedback). An unlabelled mixture may reveal that two components differ without revealing which component contains honeytokens. Common names, timestamps or certificates become transferable fingerprints only under those additional conditions.

A token registry outside every agent's view remains hidden when the pooled observations never interact with it. Communication aggregates available evidence; it does not reveal an independent secret by itself.

The 2026 incident directly establishes the communication premise. Cross-object labels require separate evidence. The Artifactory board allowed short-lived evaluation runs to pool discoveries, while the later intrusion carried state over public dead drops and recreated its tooling after local state disappeared. Repeated responses from one candidate fit Proposition~\ref{prop:pooling} directly. A naming pattern, implausible history, uniform timestamp or shared certificate across different tokens fits the cross-object interpretation only when the generator is known or prior activations and other feedback orient examples from the decoy class. Existing studies document fingerprintable artefacts in honeytokens and deployed honeypots \citep{SrinivasaEtAl2020,MsaadEtAl2023,SrinivasaEtAl2023}.

Hidden randomisation, heterogeneous implementations and per-agent decoys can reduce the transferability of prior observations when they do not introduce new stable markers. Rotation can likewise reduce the value of stale reconnaissance when the design changes faster than agents learn its invariant features \citep{ZhuangDeLoachOu2014}; visible or mechanically regular rotation may create a fingerprint. These responses may lower or shorten the information advantage. A universal guarantee against an adaptive attacker requires an additional information or authority asymmetry.

\section{The attack-detection limit of token activation}\label{sec:detection}

A honeytoken can be instrumented so that presenting its credential or requesting its URL generates an unambiguous event. This gives perfect detection of token use. Perfect detection of malicious activity with no false alarms additionally requires every successful malicious path to use the token and every benign path to avoid it. We now leave the object-classification experiment and consider a separate experiment over complete system histories.

Let $X$ denote the full history visible to a monitor: object selections, API calls, identities, timing and any other telemetry on which it may condition. Let $P_B$ be the distribution of $X$ under benign operation. An attacker chooses a strategy $\sigma$ from a nonempty set $\Sigma_A$, inducing a distribution $P_{A,\sigma}$ over the same histories. A detector is a measurable function $\delta(X)\in[0,1]$. Its strategy-specific true-positive rate and its false-positive rate are
\[
 \operatorname{TPR}_{\sigma}(\delta)=\E_{A,\sigma}[\delta(X)]
 \quad\text{and}\quad
 \operatorname{FPR}(\delta)=\E_B[\delta(X)].
\]

\begin{proposition}[The detection boundary]\label{prop:detection}
For every detector $\delta$ and attack strategy $\sigma$,
\[
 \operatorname{TPR}_{\sigma}(\delta)-\operatorname{FPR}(\delta)
 \leq \TV(P_{A,\sigma},P_B).
\]
Consequently, the robust separation value
\[
 V\triangleq\sup_{\delta}\inf_{\sigma\in\Sigma_A}
 \left\{\operatorname{TPR}_{\sigma}(\delta)-\operatorname{FPR}(\delta)\right\}
\]
satisfies
\[
 V
 \leq
 \inf_{\sigma\in\Sigma_A}\TV(P_{A,\sigma},P_B).
\]
For a fixed $\sigma$, perfect detection with no false alarms is possible if and only if $P_{A,\sigma}$ and $P_B$ are mutually singular. If some strategy $\sigma_0$ induces $P_{A,\sigma_0}=P_B$, then the robust separation value is zero and every detector has $\operatorname{TPR}_{\sigma_0}=\operatorname{FPR}$.
\end{proposition}

\noindent The bound limits inference from behaviour while leaving cryptographic verification intact. A signed log can prove that a credential was presented. Hidden intent remains uncertifiable on a history compatible with both benign and attack regimes. Overlap that prevents mutual singularity rules out perfect separation; equality of the full distributions produces the stronger conclusion that true- and false-positive rates coincide. Likewise, a trusted-agent identifier verifies which identity acted while leaving the identity's compromise status unresolved.

There is a useful path formulation. Let $\mathcal E$ be a set of context-labelled events, let $\mathcal D\subseteq\mathcal E$ contain the events the defender can instrument, and let $\mathcal B$ and $\mathcal A$ be the benign and successful attack paths. For a static trap set $K\subseteq\mathcal D$, the detector $\delta_K(x)=\mathbf 1\{x\cap K\neq\varnothing\}$ has zero benign interference if it meets no path in $\mathcal B$, and complete coverage if it meets every path in $\mathcal A$. The detector $\delta_K$ specialises the detector in Proposition~\ref{prop:detection} to membership in a trap set. Requiring zero interference and complete coverage for every path gives the worst-case, set-valued analogue of that proposition's probabilistic separation problem. Write $U_B=\cup_{b\in\mathcal B}b$ for the events used on some benign path.

\begin{corollary}[The hitting-set condition]\label{cor:hitting}
A deployable zero-interference trap set with complete coverage exists if and only if
\[
 a\cap(\mathcal D\setminus U_B)\neq\varnothing
 \qquad\text{for every }a\in\mathcal A.
\]
\end{corollary}

\noindent The cleanest honeytoken, therefore, attaches to an action with zero legitimate use. An unused administrator credential, a database row no application should request or a URL absent from every authorised workflow can supply a high-confidence signal. A decoy document embedded in an ordinary retrieval corpus cannot inherit that confidence merely by being called a honeytoken. With enough benign retrievals, even a small touch probability produces many investigations and dilutes the posterior meaning of an alert, the base-rate problem emphasised by \citet{Axelsson2000}.

Complete coverage is possible in a restricted threat model. If every successful attack contains at least one deployable event that no benign path uses, the defender can instrument a hitting set of such events, for example $\mathcal D\setminus U_B$. A stronger special case forces every path through one independently enforced gate. Complete mediation then supplies the security guarantee. The decoy may make the violation more attractive or informative, while independent enforcement makes the gate unavoidable.

The path result holds the families $\mathcal A$ and $\mathcal B$ fixed. If placing $K$ changes attacker routes or legitimate workflows, those induced path families must be recomputed. Corollary~\ref{cor:hitting} therefore characterises static support. An equilibrium trap-placement analysis must also account for induced route and workflow changes.

\section{The private-control-plane construction}\label{sec:architecture}

The common-information result points to a separation. This construction places token identity and execution authority outside the agent plane. Consider a credential broker outside the agent's readable environment. The broker stores the registry that distinguishes real credentials from honeytokens and issues a short-lived, task-scoped capability after checking provenance and authority. Protected services validate that capability at the action boundary. Trusted agents follow a simple rule: never use a credential discovered in source code, a document, a model context or shared memory; request the required capability from the broker. A planted credential found in any of those places is unnecessary to legitimate work. Figure~\ref{fig:broker} traces the resulting authority and alert paths.

This architecture separates three facts that are often conflated:
\begin{enumerate}
    \item The defender can recognise the token. The mapping remains exclusively in the broker or honeychecker, beyond agents' view.
    \item The trusted agent can avoid the token. It rejects all unbrokered credentials by provenance, without classifying individual strings as real or false.
    \item The attacker may still select the token. An agent that departs from the provenance rule and presents a discovered credential generates the alert.
\end{enumerate}
\FloatBarrier

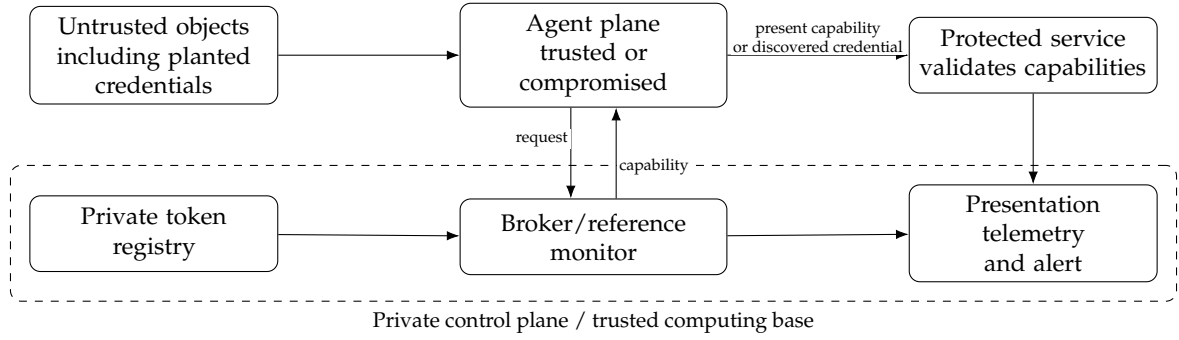
\begin{figure}[!htbp]
\centering
\begin{tikzpicture}[
    box/.style={draw, rounded corners, align=center, text width=3.45cm, minimum height=1cm, inner sep=4pt, font=\footnotesize},
    arrowlabel/.style={fill=white, inner sep=1pt, align=center, font=\scriptsize},
    >=Latex
]
\node[box, text width=3cm] (inputs) {Untrusted objects\\including planted credentials};
\node[box, text width=3.25cm, right=24mm of inputs] (agent) {Agent plane\\trusted or compromised};
\node[box, text width=3cm, right=24mm of agent] (service) {Protected service\\validates capabilities};
\node[box, text width=3cm, below=12mm of inputs] (registry) {Private token registry};
\node[box, text width=3.25cm, below=12mm of agent] (broker) {Broker/reference monitor};
\node[box, text width=3cm, below=12mm of service] (alert) {Presentation telemetry\\and alert};
\node[draw, dashed, rounded corners, fit=(registry)(broker)(alert), inner sep=2.5mm,
      label={[font=\scriptsize]below:Private control plane / trusted computing base}] {};

\draw[-{Latex}] (inputs) -- (agent);
\draw[-{Latex}] (agent) -- node[arrowlabel, above, font=\tiny] {present capability\\or discovered credential} (service);
\draw[-{Latex}] ([xshift=-3mm]agent.south) -- node[arrowlabel, left, pos=0.35, font=\tiny] {request} ([xshift=-3mm]broker.north);
\draw[-{Latex}] ([xshift=3mm]broker.north) -- node[arrowlabel, right, pos=0.35, font=\tiny] {capability} ([xshift=3mm]agent.south);
\draw[-{Latex}] (registry) -- (broker);
\draw[-{Latex}] (broker) -- (alert);
\draw[-{Latex}] (service) -- (alert);
\end{tikzpicture}
\caption{Brokered provenance. The agent presents broker-issued capabilities or discovered credentials to the protected service. The registry and monitor remain in the control plane; a planted unbrokered credential triggers an alert.}
\label{fig:broker}
\end{figure}
\FloatBarrier

\paragraph{Brokered-provenance design guarantee.}
Suppose that (i) the broker, token registry and monitor are outside the agent plane and cannot be read or altered by agents; (ii) every capability used in compliant operation is broker-issued; (iii) compliant agents reject all unbrokered credentials; (iv) each planted credential is unbrokered and has no legitimate use; and (v) the monitor observes every presentation of a planted credential and alerts reliably. Then compliant operation activates no planted credential, and every path that presents one is detected with probability one.

The broker, token registry and monitor are therefore part of the trusted computing base. Compromise of any one of them removes the guarantee. The design resembles the honeychecker in \citet{JuelsRivest2013} and the more general reference-monitor construction in \citet{Juels2014}. A stolen password file contains several plausible values, and only a separate service knows which one is genuine. No public label encodes the distinction for the attacker to inspect. For agent systems, the protected component must also be non-exportable: if the broker returns a durable secret into a compromised agent's context, the asymmetry lasts only until that secret is copied to shared memory.

The construction does not prevent malicious use of a legitimately issued capability, and it does not detect an attacker that avoids planted credentials. An attacker can also steal a broker-issued capability after issuance. The honeytoken detects one specified deviation: attempted use of an object that policy-compliant agents never use.
\begin{enumerate}
    \item \term{Keep identifiers out of the agent plane.} Token registries, callbacks and correlation keys belong in a control plane that agents cannot read. Obscure metadata visible to a trusted agent is visible to a compromised one too.
    \item \term{Route by provenance.} Legitimate agents should obtain capabilities from a broker and reject whole untrusted sources. No particular bait credential need be disclosed to them.
    \item \term{Use task-scoped, short-lived and preferably non-exportable authority.} Identity is insufficient when the authenticated agent is compromised. Authority should be checked at the action boundary and limited to the current task.
    \item \term{Personalise and refresh deceptive objects carefully.} Per-agent or per-session tokens can limit how far one observation generalises when personalisation and refreshment do not create their own stable markers.
    \item \term{Make alerts quiet and responses fast when responses are observable.} A conspicuous block or mechanically timed rotation can teach the coalition. Silent telemetry can preserve uncertainty, while rapid containment matters when it can outrun information sharing.
    \item \term{Measure contact, bypass, activation and compromise separately.} Silence may mean no attack, successful deterrence, learned avoidance or an attack along another path. A token count alone cannot distinguish them.
\end{enumerate}
Honeytokens are sensors whose precision can be exceptionally high when legitimate use is structurally zero. In strategic environments studied in the deception literature, they may impose search costs or deter an attacker who remains uncertain \citep{PibilEtAl2012,KiekintveldEtAl2015,CrompEtAl2023}. Those strategic effects lie outside this feasibility model, whose guarantee is conditional on activation. A complete partition between benign and malicious behaviour still requires a security boundary.

\section{Conclusion}\label{sec:conclusion}

The Hugging Face incident revealed two forms of persistence. Before the external intrusion, agents used shared Artifactory state to carry exploit knowledge across otherwise separate evaluation runs. During the intrusion, public services and attacker-controlled datasets carried commands and results across fresh sandboxes. The incident therefore changes the relevant unit of analysis: even when honeytokens continue to work, separate agents sharing a memory substrate form one learning adversary.

When a compromised agent shares the trusted agent's information and can implement its policy, a trusted-agent indicator is an attacker indicator. If the defender removes every visible difference, trusted agents lose the same ability to distinguish genuine from false objects. Repeated probes make a persistent difference easier to learn for a fixed candidate when the coalition remains active long enough to collect them; transfer across candidates additionally requires a stable deployment rule and prior knowledge or feedback that orients the classes. At the system level, a perfectly instrumented token certifies its own activation and leaves every malicious path that avoids the token uncovered.

With common information and a copyable trusted policy, durable asymmetry requires protected information or authority outside the agent plane; the broker supplies one construction. Universal detection further requires every successful attack path to cross some deployable monitored event outside benign support.

\clearpage
\appendix
\section{Proofs}\label{app:proofs}

The first two proofs concern object selection under common information, and the third concerns repeated probes of a fixed candidate. Proposition~\ref{prop:detection} and Corollary~\ref{cor:hitting} then turn to post-interaction detection. The arguments use policy copying, the layer-cake representation, contraction of total variation, Hoeffding's inequality, and measurable or set-theoretic separation.

\begin{proof}[Proof of Proposition~\ref{prop:copy}]
Let $a$ denote the attacker's policy. By assumption, the trusted policy $b$ is available to the attacker. Set $a=b$. By $(\alpha,\beta)$-compatibility, the copied policy's genuine-object use rate satisfies
\[
 \E_R[a(Z)]=\E_R[b(Z)]\geq 1-\alpha,
\]
while its honeytoken-activation rate satisfies
\[
 \E_H[a(Z)]=\E_H[b(Z)]\leq\beta.
\]
Thus the infimum activation rate over attack policies attaining a genuine-object use rate of at least $1-\alpha$ is no greater than $\beta$. If $\alpha=\beta=0$, the same copied policy uses genuine objects with probability one and honeytokens with probability zero, proving the final statement.
\end{proof}

\begin{proof}[Proof of Proposition~\ref{prop:tv}]
Fix a measurable policy $q:\mathcal Z\rightarrow[0,1]$. For each $z\in\mathcal Z$, the layer-cake representation gives
\[
 q(z)=\int_0^1 \mathbf 1\{q(z)>t\}\,dt.
\]
Taking expectations under each type and applying Tonelli's theorem yields
\[
 \E_Rq-\E_Hq
 =\int_0^1\bigl[P_R(q>t)-P_H(q>t)\bigr]\,dt.
\]
The triangle inequality and the definition of total variation therefore give
\begin{align*}
 \left\lvert\E_Rq-\E_Hq\right\rvert
 &\leq \int_0^1
 \left\lvert P_R(q>t)-P_H(q>t)\right\rvert\,dt\\
 &\leq \TV(P_R,P_H).
\end{align*}
For an $(\alpha,\beta)$-compatible policy $b$, the two compatibility inequalities imply
\[
 \E_Rb-\E_Hb\geq 1-\alpha-\beta,
\]
and the inequality just proved bounds the left-hand side above by $\TV(P_R,P_H)$. Combining this inequality with the non-negativity of total variation gives the stated maximum with zero.

For the classification statement, view $q$ as a randomised rule that classifies the object as genuine with probability $q(Z)$. Using the bound already established, its equal-prior error satisfies
\[
 e(q)=\frac{1}{2}\left[(1-\E_Rq)+\E_Hq\right]
 \geq \frac{1-\TV(P_R,P_H)}{2}.
\]
To attain this lower bound, let $\mu=P_R+P_H$, which dominates both probability measures, and choose
\[
 q^*=\mathbf 1\{dP_R/d\mu\geq dP_H/d\mu\}.
\]
By the positive-set characterisation of total variation, this likelihood-ratio region attains $\TV(P_R,P_H)$. Substitution into the error expression therefore gives equality and proves the stated Bayes-error formula.
\end{proof}

\begin{proof}[Proof of Proposition~\ref{prop:pooling}]
\emph{Monotonicity.}
Let $\pi_n$ project $Z^{n+1}$ onto its first $n$ coordinates. Projective consistency implies $P_\theta^n=P_\theta^{n+1}\circ\pi_n^{-1}$ for each $\theta$. Total variation contracts under measurable maps, so
\[
 \Delta_n
 =\TV(P_R^{n+1}\circ\pi_n^{-1},P_H^{n+1}\circ\pi_n^{-1})
 \leq \TV(P_R^{n+1},P_H^{n+1})
 =\Delta_{n+1}.
\]

\emph{Convergence.}
Because the one-observation distributions differ, there is a measurable set $A$ for which $p_R=P_R(A)\neq P_H(A)=p_H$. Interchanging the type labels if necessary, take $p_R>p_H$, and let $c=(p_R+p_H)/2$. Define
\[
 B_n=\left\{\frac{1}{n}\sum_{i=1}^n\mathbf 1\{Z_i\in A\}>c\right\}.
\]
Write $\gamma=p_R-p_H>0$. Under the conditionally i.i.d. assumption, Hoeffding's inequality gives
\[
 P_R^n(B_n^c)\leq \exp\left(-\frac{n\gamma^2}{2}\right)
 \quad\text{and}\quad
 P_H^n(B_n)\leq \exp\left(-\frac{n\gamma^2}{2}\right).
\]
The equal-prior error of this empirical-frequency test is therefore at most $\exp(-n\gamma^2/2)$. The Bayes error $e_n^*$ is no larger, and Proposition~\ref{prop:tv}'s testing identity gives
\[
 \Delta_n=1-2e_n^*
 \geq 1-2\exp\left(-\frac{n\gamma^2}{2}\right).
\]
The right-hand side converges to one, while $\Delta_n\leq1$, so $\Delta_n\rightarrow1$ and $e_n^*\rightarrow0$.
\end{proof}

\begin{proof}[Proof of Proposition~\ref{prop:detection}]
Fix a detector $\delta$ and an attack strategy $\sigma$. Applying Proposition~\ref{prop:tv}'s first inequality to this bounded function and the distributions $P_{A,\sigma}$ and $P_B$ gives
\[
 \left\lvert\E_{A,\sigma}\delta-\E_B\delta\right\rvert
 \leq\TV(P_{A,\sigma},P_B).
\]
Dropping the absolute value in the direction required by the proposition yields
\[
 \operatorname{TPR}_{\sigma}(\delta)-\operatorname{FPR}(\delta)
 =\E_{A,\sigma}\delta-\E_B\delta
 \leq\TV(P_{A,\sigma},P_B).
\]
Because the inequality holds pointwise in $\sigma$, every detector satisfies
\[
 \inf_{\sigma\in\Sigma_A}
 \left\{\operatorname{TPR}_{\sigma}(\delta)-\operatorname{FPR}(\delta)\right\}
 \leq
 \inf_{\sigma\in\Sigma_A}\TV(P_{A,\sigma},P_B).
\]
The right-hand side does not depend on $\delta$, so taking the supremum over detectors proves the robust bound.

For a fixed $\sigma$, if detection is perfect and produces no false alarms, then the left-hand side equals one. Because total variation between probability measures cannot exceed one, the bound forces $\TV(P_{A,\sigma},P_B)=1$. This equality holds exactly when $P_{A,\sigma}$ and $P_B$ are mutually singular.

Conversely, mutual singularity provides a measurable set $S$ with $P_{A,\sigma}(S)=1$ and $P_B(S)=0$. The detector $\delta(x)=\mathbf 1\{x\in S\}$ therefore has true-positive rate one and false-positive rate zero. Finally, if $P_{A,\sigma_0}=P_B$ for some $\sigma_0$, then $\E_{A,\sigma_0}\delta=\E_B\delta$ for every detector. Hence every detector has $\operatorname{TPR}_{\sigma_0}=\operatorname{FPR}$, which bounds robust separation above by zero. The constant detector $\delta\equiv0$ attains zero against every strategy and supplies the matching lower bound. Thus the robust separation value equals zero.
\end{proof}

\begin{proof}[Proof of Corollary~\ref{cor:hitting}]
\emph{Necessity.}
Suppose a deployable zero-interference, complete-coverage set $K$ exists. Because $K\subseteq\mathcal D$, zero interference gives $K\subseteq\mathcal D\setminus U_B$. Now fix an attack path $a\in\mathcal A$. Complete coverage gives an event in $K\cap a$, and that event belongs to $a\cap(\mathcal D\setminus U_B)$. Thus every attack path contains a deployable event outside $U_B$.

\emph{Sufficiency.}
Conversely, suppose $a\cap(\mathcal D\setminus U_B)\neq\varnothing$ for every $a\in\mathcal A$, and set $K=\mathcal D\setminus U_B$. Every benign path is a subset of $U_B$, so no benign path meets $K$. The maintained condition gives $a\cap K\neq\varnothing$ for every attack path $a\in\mathcal A$. Hence $K$ is deployable and has zero benign interference and complete coverage.
\end{proof}

\newpage

\bibliographystyle{apalike}
\bibliography{references-agentstalk}

\end{document}